\documentclass[11pt,a4paper]{article}

\usepackage{fullpage}

\usepackage{color}
\usepackage{graphicx}
\usepackage{amsmath}
\usepackage{amssymb}
\usepackage{amsthm}
\usepackage[noend]{algpseudocode}
\usepackage{algorithm}
\usepackage{subcaption}
\usepackage{url}
\usepackage{hyperref}

\graphicspath{{./figures/}}

\newcounter{theorem}
\newtheorem{lemma}{Lemma}
\newtheorem{theorem}[lemma]{Theorem}

\newtheorem{proposition}[lemma]{Proposition}
\newtheorem{corollary}[lemma]{Corollary}

\let\geq\geqslant
\let\leq\leqslant
\newcommand{\IR}{\mathbb{R}}
\newcommand{\IZ}{\mathbb{Z}}

\newcommand{\cell}{\mathcal{C}}
\newcommand{\poincare}{\mathbb{H}}
\newcommand{\discrete}{\mathbb{B}}
\newcommand{\level}{\mathrm{lev}}
\newcommand{\rg}{\mathrm{rg}}

\DeclareMathOperator{\arsinh}{arsinh}

\newcounter{casenum}

\title{Embeddings and near-neighbor searching with constant additive error for hyperbolic spaces
\thanks{
This work was supported by the National Research Foundation of Korea (NRF) grant
funded by the Korea government (MSIT) (No. 2022R1F1A107586911).
}
}

\author{Eunku Park\\
School of Electrical and Computer Engineering\\
UNIST, Republic of Korea\\
\texttt{parkeun9@unist.ac.kr}
\and
Antoine Vigneron\thanks{Corresponding author}\\ 
School of Electrical and Computer Engineering\\ 
UNIST, Republic of Korea\\
\texttt{antoine@unist.ac.kr}
}

\begin{document}
\maketitle
\begin{abstract}
	We give an embedding of the Poincar\'e halfspace $\poincare^D$ into a discrete metric space
	based on a binary tiling of $\poincare^D$, with additive distortion $O(\log D)$.
	It yields the following results.
	We show that any subset $P$ of $n$ points in $\poincare^D$ 
	can be embedded into a graph-metric with $2^{O(D)}n$ vertices and edges, and with additive
	distortion $O(\log D)$. We also show how to construct, for any $k$, an $O(k\log D)$-purely 
	additive spanner of $P$ with $2^{O(D)}n$ Steiner	vertices and $2^{O(D)}n \cdot \lambda_k(n)$ edges, where 
	$\lambda_k(n)$ is the $k$th-row inverse Ackermann function. 
	Finally, we show how to construct an approximate Voronoi diagram for $P$ of size $2^{O(D)}n$.
	It allows us	to answer approximate near-neighbor queries in $2^{O(D)}+O(\log n)$ time,
	with additive error $O(\log D)$.
	These constructions can be done in $2^{O(D)}n \log n$ time.
\end{abstract}


\section{Introduction}

The Poincar\'e halfplane $\poincare^2$ is perhaps the most common model of hyperbolic spaces, together with the
Poincar\'e disk which it is isometric to. The points of $\poincare^2$ are the points $(x,z)$, $z>0$ of
the upper halfplane of $\IR^2$. The shortest paths in $\poincare^2$, called geodesics,
are arcs of circles orthogonal to
the $x$-axis (\figurename~\ref{fig:models}a) and the arc length is given by integrating 
the relation $ds^2=(dx^2+dz^2)/z^2$. More generally, the Poincar\'e halfspace is a $D$-dimensional
model of hyperbolic space, that consists of points $(x,z)$ where $x \in \IR^{D-1}$ and $z$ is a positive
real number. In $\poincare^D$, the expression of $ds^2$ is the same,
but $x$ is now a point in $\IR^{D-1}$, and geodesics are arcs of circles orthogonal to the hyperplane $z=0$.

Hyperbolic spaces behave very differently from Euclidean spaces in some respects. For instance, in fixed
dimension, the volume of a hyperbolic ball grows exponentially with its radius, while the radius
of a Euclidean ball grows polynomially. A triangle, formed by connecting three
points by the geodesic between each pair of these points, is thin in the sense that from any point
on an edge, there is a point on another edge at distance bounded by a constant.

As a consequence, hyperbolic spaces are sometimes more suitable than Euclidean spaces to represent 
some types of data. It has been shown, for instance, that there are better embeddings of the internet 
graph into  hyperbolic spaces, compared with its embeddings into Euclidean spaces~\cite{1354510}. There has also been recent interest
in hyperbolic spaces in the context of artificial neural networks~\cite{GaneaBH18}.

In this paper, we present embeddings of finite subsets of $\poincare^D$ into graph metrics
with a linear number of edges, and additive distortion $O(\log D)$.
As an application, we present an approximate near-neighbor data structure with 
$O(\log D)$ additive distortion. These two results have no multiplicative distortion.

\begin{figure}
	\centering
	\includegraphics[width=\textwidth]{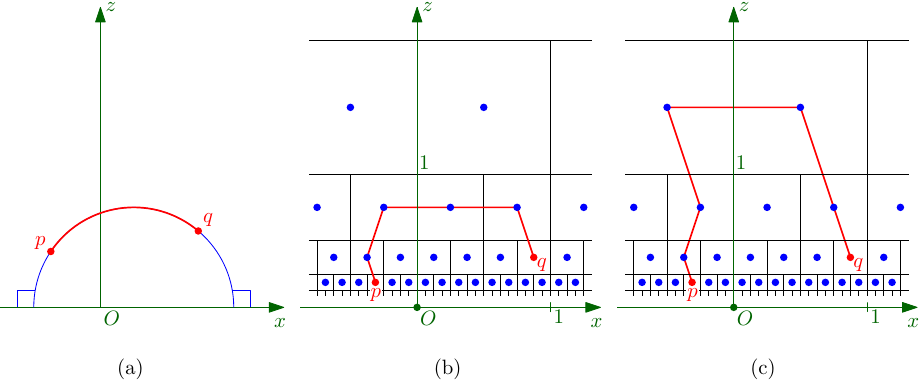}
	\caption{Three models of hyperbolic spaces.
	(a) The Poincar\'e halfplane $\poincare^2$. 
	(b) The first discrete model $(\discrete^2,d_1)$, with a shortest path of length
		$d_1(p,q)=5$.
	(c) The second discrete model $(\discrete^2,d_2)$, with a shortest path of length
		$d_2(p,q)=6$.
	\label{fig:models}}
\end{figure}

\subsection{Our results.} 
Given two metric spaces $(M,d)$ and $(M',d')$, we say that a mapping $\varphi:M \to M'$
is an {\it embedding with additive distortion}  $\Delta$
if for any two points $p,q \in M$, we have
\[
	d(p,q) -\Delta \leq d'(\varphi(p),\varphi(q)) \leq d(p,q)+\Delta.
\]

Our first result (Theorem~\ref{th:embedmain}) is an embedding of $\poincare^D$ with additive distortion
$O(\log D)$ into a discrete metric space $(\discrete^D,d_1)$ that is obtained
from a binary tiling of $\poincare^D$. In two dimensions, this binary tiling uses boxes whose sizes increase
exponentially with the $z$-coordinate. (See \figurename~\ref{fig:models}b.). A point
is placed at the center of each box, and the distance between two points $p$ and $q$ is the minimum number
of box boundaries that are crossed when going from $p$ to $q$.
(See Section~\ref{sec:models} for a detailed description.)

Given a subset $P$ of $n$ points of a metric space $(M,d)$, a {\it purely-additive spanner} of $P$
with distortion $\Delta$ and $n'$ {\it Steiner points} is a graph $G(V,E)$ where $S\subset M$, $|S|=n'$,
the points in $S$ are called Steiner points, $V=P \cup S$,  the length of any edge $pq$ is $d(p,q)$, and the
shortest path distance $d_G(\cdot)$ in this graph satisfies
\[
	d(p,q) \leq d_G(p,q) \leq d(p,q)+\Delta \quad \text{for all }p,q\in P.
\]

We show, for any subset $P$ of $n$ points of $(\discrete^D,d_1)$, 
how to construct in $2^{O(D)} n \log n$ time
a purely-additive spanner of $P$ with $2^{O(D)}n$ edges and Steiner vertices, and
distortion $2$. (Theorem~\ref{th:spanner_d1}.)

Based on the two results above, we first obtain an embedding of any subset $P$ of $n$ points
of $\poincare^D$ into a graph metric with $2^{O(D)}n$ vertices and edges, and
additive distortion $O(\log D)$. We also obtain an $O(k\log D)$ purely additive spanner of $P$
with $2^{O(D)}n$ vertices and $2^{O(D)}n  \cdot \lambda_k(n)$ Steiner vertices and edges, where $\lambda_k(n)$
is the $k$th-row inverse Ackermann function. 

Given a subset $P$ of $n$ points of a metric space $(M,d)$, and a query point $q \in P$, 
a {\it nearest neighbor} of $q$ is a point $p \in P$ such that $d(p,q)$ is minimum.
An {\it approximate nearest neighbor} (ANN) with additive distortion $\Delta$ is a point
$p' \in P$ such that $d(p',q) \leq d(p,q) +\Delta$. 

We give data structures for
answering ANN queries in $(\discrete^D,d_1)$ and in $\poincare^D$ with 
query time $2^{O(D)}+O(\log n)$ and construction time $2^{O(D)} n\log n$. For $(\discrete^D,d_1)$,
the additive distortion is 2, and for $\poincare^D$, it is $O(\log D)$. 
(See Theorem~\ref{th:nn2} and Corollary~\ref{cor:nnh}.) 
These data structures are in fact Approximate Voronoi Diagrams (AVD): They give a partition
of the space into $2^{O(D)}n$ regions, each region being associated with $2^{O(D)}$ representative points,
and such that for any query point in one of these regions, one of the representative points 
 is an approximate near neighbor.

\subsection{Comparison with previous work.} Spanners have been studied in the more general context
of an arbitrary weighted graph metric~\cite{AlthoferDDJS93}. For instance, one can find a $t$-multiplicative
spanner of total weight $(1+1/t)$ times the weight of a minimum spanning tree.

For Euclidean metrics, it was shown that there are spanners with a linear number
of edges, and multiplicative distortion arbitrarily close to 1~\cite{AryaDMSS95}. One difference with
this paper is that we consider an {\it additive} distortion. In the worst case, one cannot hope to 
find a non-trivial additive spanner for the Euclidean metric, because the additive error for a given
graph can be made arbitrarily larger by scaling up the input point set.

Gromov-hyperbolicity is a notion of hyperbolicity that applies to any metric space $(M,d)$,
including discrete metric spaces. These spaces have the property that for any 4 points $p,q,r,s \in M$,
the largest two sums of distances among $d(p,q)+d(r,s)$, $d(p,r)+d(q,s)$, $d(p,s)+d(q,r)$ differ 
by at most a constant $\delta$. 
As shown by Gromov~\cite{Gromov1987}, any metric space of $n$ points with hyperbolicity $\delta$ can be
embedded into a tree-metric with additive distortion $O(\delta \log n)$. The Poincar\'e half-space
has hyperbolicity $O(1)$, so this result applies to a more general type of hyperbolic spaces 
than our embedding of $\poincare^D$ 
into a graph metric with $O(n)$ edges. On the other hand, we obtain an additive distortion $O(1)$
in constant dimension $D=O(1)$, while Gromov's construction gives $O(\log n)$ in our case.
Chepoi et al.~\cite{ChepoiDEHVX12} give additive $O(\delta \log n)$ spanners with $O(\delta n)$ edges
for unit graph metrics (i.e. metrics for graphs where each edge weight is equal to 1) that are
$\delta$-hyperbolic. Compared with our result, it allows arbitrary hyperbolicity, but it is restricted
to unit graphs, and the distortion is logarithmic.
Lee and Krautghamer~\cite{KL06} present multiplicative spanners with $O(n)$ edges for locally doubling,
geodesic Gromov-hyperbolic spaces.

There has also been some work on problems other than spanners and embeddings in hyperbolic spaces. 
Lee and Krautghamer~\cite{KL06} presented an ANN data structure for
geodesic and locally doubling Gromov-hyperbolic spaces, with $O(1)$ additive error, 
$O(n^2)$ space usage and $O(\log^2 n)$ query time.
Kisfaludi-Bak et al. gave an algorithm for the TSP problem in $\poincare^2$~\cite{kisfaludibak}.

More recently, Kisfaludi-Bak and  van Wordragen~\cite{kisfaludibak2023quadtree} gave an ANN data
structure and Steiner spanners for $\poincare^D$ with {\it multiplicative} error $(1+\varepsilon)$.
Their approach is based on the same binary tiling that we use, but they derive from it
a non-trivial type of quadtree that is taylored for providing multiplicative guarantees in 
hyperbolic spaces.  The lower bound
that they present on spanners without Steiner points imply that in order to an achieve additive error
as we do, Steiner points are also required.

\subsection{Our approach.} 
We use an approximation of the Poincar\'e halfspace $(\poincare^D,d_H)$
by a discrete metric space $(\discrete^D,d_1)$ where the points are centers of hypercubes (called {\it cells}) 
whose sizes grow exponentially with $y$. (See \figurename~\ref{fig:models}b.) 
The distance between two points $p,q$ is the minimum number
of cells that are crossed when going from $p$ to $q$.
This discrete model was  mentioned,
for instance, by Cannon et al.~\cite{Cannon97}. 
We present these models of hyperbolic spaces in Section~\ref{sec:models}.

We introduce a different distance function $d_2$ over $\discrete^D$: we go upwards from $p$ and $q$
until we reach adjacent squares, and then we connect the subpaths using a single horizontal edge.
(See \figurename~\ref{fig:models}c.).
This model is not  a metric as it does not satisfy the triangle inequality. 
In Section~\ref{sec:spdiscrete},  we show that $d_1$ and $d_2$ differ by at most 2.  
Then in Section~\ref{sec:embedding}, we show that over $\discrete^D$, the metric $d_H$ 
differs from $d_1$ and $d_2$ by an additive error $O(\log D)$.

These discrete models lend themselves well to the use of quadtrees, and the efficient constructions 
of our embeddings, spanners and ANN-data structure, as well as the bound on their size, are based 
on compressed quadtrees.
In Section~\ref{sec:quadtree}, we present the compressed quadtree data structure that we use.

This structure allows us, in Section~\ref{sec:spannerdiscrete}, to show that the
overlay of the shortest paths according to $d_2$ has linear complexity, and hence gives a 2-additive spanner of
linear size with respect to $d_1$. (See \figurename~\ref{fig:size}.)
In Section~\ref{sec:hspanner}, we use the spanner 
we constructed in the discrete model to obtain an embedding  with $O(\log D)$ additive 
distortion into the metric of a graph that has $2^{O(D)}n$ vertices and edges.
Then using a transitive-closure spanner, we turn it into a spanner with respect to $d_H$. 
This spanner is embedded in $\poincare^D$, and has $2^{O(D)}n$ edges and Steiner points.

Our data structure for approximate near neighbor (ANN) searching is presented in Section~\ref{sec:NN}.
This data structure is an approximate Voronoi diagram for the discrete model $\discrete^D$, 
that is based on compressed quadtrees. This construction is done in two steps. We first compute the
subdivision induced by the (minimal) quadtree that records the input point set $P$. Then we refine this
subdivision by adding additional quadtree boxes where it is needed.


\section{Models of hyperbolic spaces}\label{sec:models}

\figurename~\ref{fig:models} shows the three models of hyperbolic spaces that we consider
in this paper, when $D=2$.
The  Poincar\'e halfplane $(\poincare^2,d_H)$ is shown in \figurename~\ref{fig:models}a.
It is a model of 2-dimensional hyperbolic space with constant negative curvature $-1$.
The halfplane $\poincare^2$ consists of the points $p=(x(p),z(p))$ where $z(p)>0$. 
The geodesics in $(\poincare^2,d_H)$ are arcs of semi-circles that are orthogonal to the $x$-axis. 
(See \figurename~\ref{fig:models}a.)

More generally, when $D \geq 2$, the Poincar\'e halfspace $\poincare^D$ consists of the points 
$p=(x(p),z(p))$ where $x\in \IR^{D-1}$ and $z$ is a positive real number.
The distance $d_H(p,q)$ between two points
$p$ and $q$ is the length of the  geodesic from $p$ to
$q$, where the arc-length is given by the relation $ds^2=(dx^2+dz^2)/z^2$. 
This distance is given by the expression
\[
	d_H(p,q)=2\arsinh\left( \frac 12
		{
			\frac{\|pq\|}{\sqrt{z(p)z(q)}}
		}
	\right)
\]
where $\|\cdot\|$ is the Euclidean norm.
In the special case where $x(p)=x(q)$ and $z(q)>z(p)$, it is simply $d_H(p,q)=\ln(z(q)/z(p))$. 
The geodesics in $(\poincare^D,d_H)$ are arcs of semi-circles that are orthogonal to the 
hyperplane $z=0$. 

The first discrete model $(\discrete^D,d_1)$  is defined
using a binary tiling of $\poincare^D$ with the hypercubes, called {\it cells}, 
$$\cell=[k_12^i,(k_1+1)2^i] \times [k_22^i,(k_2+1)2^i] \times  \dots \times 
[k_{D-1}2^i,(k_{D-1}+1)2^i]  \times [2^i,2^{i+1}]$$
where $k_1,k_2,\dots,k_{D-1},i \in \IZ$. 
(See  \figurename~\ref{fig:models}b for an example when $D=2$.)
The {\it level} of this cell $\cell$ is the integer $\level(\cell)=i$,
and its {\it width}  is $2^i$.

The {\it parent} of a cell $\cell$ at level $i$ is the cell at level $i+1$ whose bottom facet
contains the upper facet of $\cell$. The {\it children} of $\cell$ are the cells whose parent is $\cell$,
hence $\cell$ has $2^{D-1}$ children. An ancestor of $\cell$ is either the parent of $\cell$, or a parent of
an ancestor of $\cell$. A descendant of $\cell$ is either a child of $\cell$, or a descendant of a child
of $\cell$.
The {\it horizontal neighbors} of $\cell$ are the cells at level $i$ different from $\cell$ 
that intersect $\cell$ along its boundary, hence $\cell$ has $3^{D-1}-1$ horizontal neighbors.

We denote by $b(\cell)$ the center of the cell $\cell$, so the $z$-coordinate of $b(\cell)$ 
is $z(b(\cell))=3 \cdot 2^{i-1}$ when $\level(\cell)=i$.
Then $\discrete^D$ is the set of the points $b(\cell)$ for all the cells $\cell$.
The level $\level(b(\cell))$ of the point $b(\cell)$ is the level of $\cell$. The parent, children, 
ancestors, descendants and horizontal
neighbors of $\cell$ are the point $b(\cell')$ where $\cell'$ is a parent, child 
ancestors, descendants and horizontal neighbors of $\cell$, respectively.

From a point $p \in \discrete^D$, we allow the following types of moves:
\begin{itemize}
	\item An {\it upward move} to the parent of $p$.
	\item A {\it downward move} to a child of $p$.
	\item A {\it horizontal move} to a horizontal neighbor of $p$.
		(Hence, when moving horizontally, we allow to move along diagonals.)
\end{itemize}
(See \figurename~\ref{fig:BD}.)

\begin{figure}
	\centering
	\includegraphics[scale=.55]{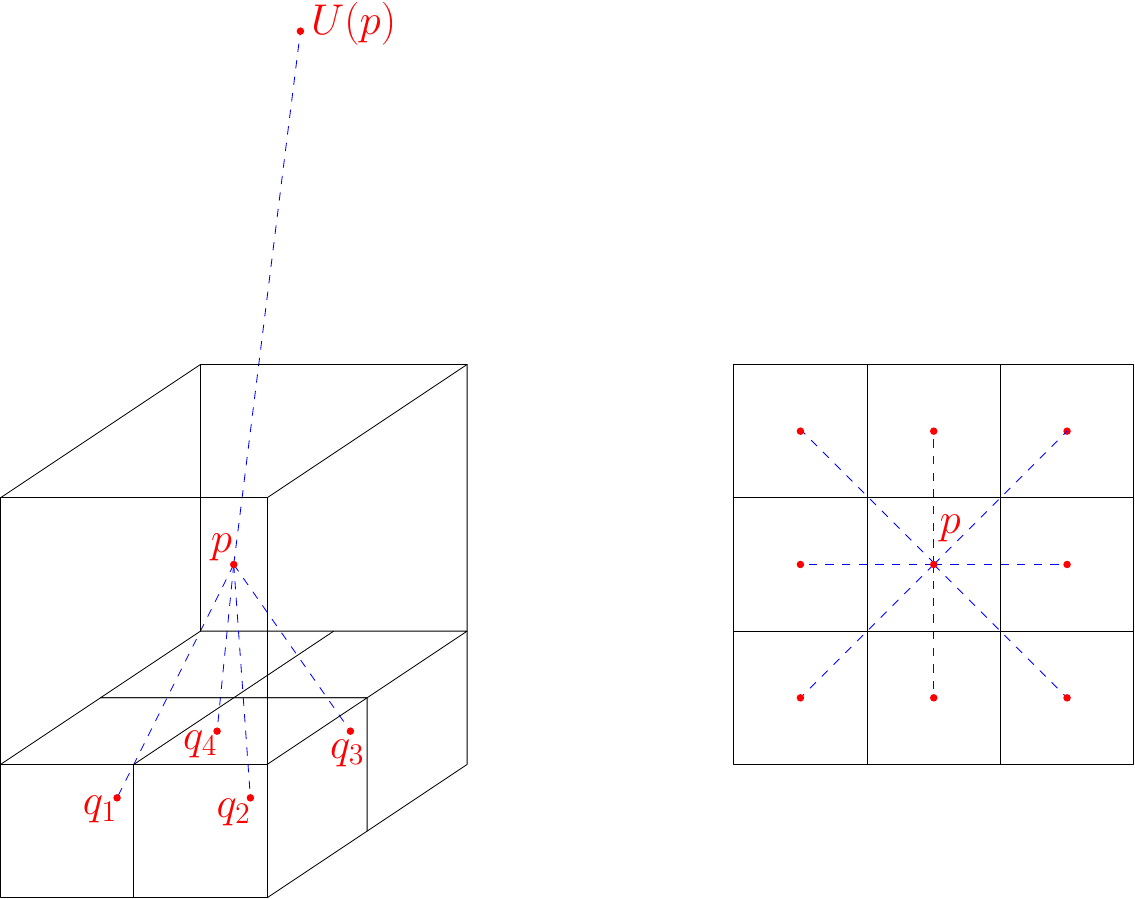}
	\caption{Possible moves in $\discrete^3$. (left) Vertical moves $U(p)$ and $D(p)\in\{q_1,q_2,q_3,q_4\}$.
		(right) The 8 Horizontal moves $H(p)$, seen from above.\label{fig:BD}}
\end{figure}

Then for any $p,q \in \discrete^D$, we define $d_1(p,q)$ to be the minimum number of moves that are needed
to reach $q$ from $p$, and we define $d_2(p,q)$ to be the length of the shortest path from $p$ to $q$ that
has at most one horizontal move. \figurename~\ref{fig:models}(b) and (c) show examples of
shortest paths in these models. A path from $p$ to $q$ consisting of $d_1(p,q)$ such moves is called
a {\it $d_1$-path}. Similarly, the path from $p$ to $q$ consisting of $d_2(p,q)$ moves,
at most one of which being horizontal, is called a {\it $d_2$-path}. For any $p,q$, there
is only one $d_2$-path from $p$ to $q$: If $p$ is neither a descendant nor an ancestor of $q$, then
this path bends at the lowest ancestors of $p$ and $q$ that are horizontal neighbors.

\begin{figure}
	\centering
	\includegraphics[scale=.7]{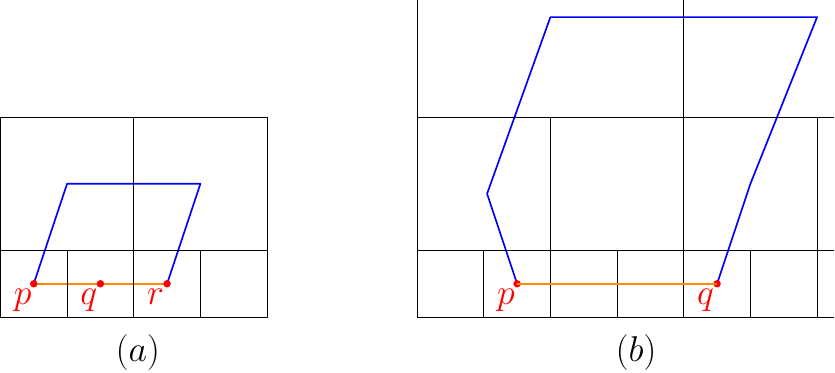}
	\caption{(a) A counterexample that shows that $d_2$ does not obey the triangle inequality:
		we have $d_2(p,q)=1$, $d_2(q,r)=1$ and $d_2(p,r)=3$
		(b) An example where $d_2(p,q)=5=2+d_1(p,q)$.
	\label{fig:moves}}
\end{figure}

The space $(\discrete^D,d_2)$ is not a metric space, but a semi-metric space, as
it does not satisfy the triangle inequality. (See counterexample in \figurename~\ref{fig:moves}a).


\section{Shortest paths in the discrete model} \label{sec:spdiscrete}

In this section, we study the structure, and give bounds on the length of $d_1$-paths
and $d_2$-paths.

\begin{lemma}\label{lem:d1shortest_prel}
	Let $p,q \in \discrete^D$, and suppose that there is a $d_1$-path from $p$ to 
	$q$ containing at least one upward move. Then there is a $d_1$-path from
	$p$ to $q$ whose first move is upward.
\end{lemma}
\begin{proof}
	Let $\rho=(p_0,p_1,\dots,p_k)$ be a path of length $k=d_1(p_0,p_k)$ from $p=p_0$ to $q=p_k$,
	containing at least one upward move. If the first move of $\rho$ is upward, then we
	are done. Otherwise,
	Let $(p_i,p_{i+1})$ be the first upward move, hence $(p_{i-1},p_i)$ is horizontal
	or downward. If $(p_{i-1},p_i)$ is downward, then we have $p_{i-1}=p_{i+1}$, so
	the path obtained from $\rho$ by
	deleting $p_{i}$ and $p_{i+1}$ is a path from $p$ to $q$ of length $d_1(p,q)-2$,
	a contradiction.
	
	Now suppose that $(p_{i-1},p_i)$ is horizontal. As $p_{i-1}$ and $p_i$ are horizontal
	neighbors, their parents either are equal, or are horizontal neighbors. So the parent of $p_{i-1}$
	is either $p_{i+1}$, or is 
	a horizontal neighbor of $p_{i+1}$. The parent of $p_{i-1}$ cannot be $p_{i+1}$,
	because if it were the case, we would obtain a path from $p$ to $q$ of length
	$d_1(p,q)-1$ by deleting $p_i$ from $\rho$. So the parent of $p_{i-1}$ is a horizontal neighbor
	$p'_i$ of $p_{i+1}$. Hence, we obtain a path $\rho'$ from $p$ to $q$ of length $d_1(p,q)$
	by replacing $p_i$ with $p'_i$. 
	
	The first upward move of $\rho'$ is now $p_{i-1}p'_i$,
	hence it has been moved one position to the left. By repeating this process $i-1$ times,
	we find a $d_1$-path from $p$ to $q$ that starts with an 
	upward move.
\end{proof}

\begin{lemma}\label{lem:d1shortest}
	For any $p,q \in \discrete^D$, there is a $d_1$-path  from $p$ to $q$ that
	consists of $m_1$ upward moves, followed by $m_2$ horizontal moves, and finally $m_3$ downward
	moves, where $m_1,m_2,m_3 \geq 0$.
\end{lemma}
\begin{proof}
	By applying Lemma~\ref{lem:d1shortest_prel} repeatedly, we obtain a $d_1$-path from $p$ to
	$q$	consisting of a sequence $\alpha$ of upward moves, followed by a sequence $\beta$ of horizontal or
	downward moves. Let $\gamma$ be the path obtained by following $\beta$ backward, hence
	$\gamma$ consists of upward or horizontal moves. We apply Lemma~\ref{lem:d1shortest_prel} 
	to $\gamma$, until we obtain a path $\gamma$ consisting of upward moves followed by
	horizontal moves. Then the path obtained by following $\alpha$, and then following $\gamma$ 
	backwards, has the desired property.
\end{proof}

Let $p,q \in \discrete^D$ be two points at the same level $i$.
Their {\it horizontal distance} $\lambda(p,q)$ is the length of a shortest path from $p$ to $q$
consisting of horizontal moves only. It is given by the expression
$\lambda(p,q)=\|x(p)x(q)\|_\infty/2^i$ where $\|\cdot\|_\infty$ is the $L^\infty$ norm.

\begin{lemma}\label{p:lambda}
	If $p$ and $q$ are two points of $\discrete^D$ and $\level(p)=\level(q)$, then 
	$2\lambda(p',q')-1 \leq  \lambda(p,q) \leq 2\lambda(p',q')+1$,
	where $p'$ and $q'$ are the parents of $p$ and $q$, respectively
\end{lemma}
\begin{proof}
	Let $\Delta_i=\{-2^{i-1},0,2^{i-1}\}^{D-1}$, then we have $\{x(p')-x(p),x(q')-x(q)\} \subset \Delta_i$.
	By the triangle inequality, it implies that
	$$\|x(p')x(q')\|_{\infty}-2^i \leq \|x(p)x(q)\|_\infty \leq \|x(p')x(q')\|_{\infty}+2^i.$$
	The result follows by dividing these inequalities by $2^i$.
\end{proof}

We now obtain a recurrence relation for $d_1(p,q)$ when $p$ and $q$ are at the same level.
\begin{lemma}\label{lem:d1samerow}
	If $p$ and $q$ are two points of $\discrete^D$ such that $\level(p)=\level(q)$, 
	and if $\lambda=\lambda(p,q)$ is their horizontal distance, then 
	$$d_1(p,q) = \begin{cases}
		\lambda & \text{ if } \lambda \leq 4, \text{ and} \\ 
		2 + d_1(p', q') & \text{ if } \lambda \geq 5. 
		\end{cases}$$
		where $p'$ and $q'$ are the respective parents of $p$ and $q$.
\end{lemma}
\begin{proof}
	By Lemma~\ref{lem:d1shortest}, there is a shortest path from $p$ to $q$
	that either consists of horizontal moves only (type 1), or that  goes through
	$p'$ and $q'$ (type 2).
	
	Any path of type 2 has length at least 3. So if $\lambda \leq 3$, there
	is a shortest path of type 1, and thus $d_1(p,q)=\lambda$.
	
	If $\lambda=4$, then by Lemma~\ref{p:lambda}, we have $\lambda(p',q') \geq 2$,
	so any path of type 2 has length more than 3. So there is a shortest path 
	of type 1, and thus $d_1(p,q)=\lambda$.
	
	Now suppose that $\lambda \geq 5$. Then by Lemma~\ref{p:lambda}, we have
	$\lambda(p',q') \leq (1+\lambda)/2 \leq \lambda-2$, and thus there is a shortest
	path which is of type 2. It follows that $d_1(p,q)=2+d_1(p',q')$.
\end{proof}

\begin{lemma} \label{lem:ineqsamerow}
	If $p$ and $q$ are two points of $\discrete^D$ and $\level(p)=\level(q)$, then
	$d_2(p,q) \leq d_1(p,q)+2$. 
\end{lemma} 
\begin{proof}
	Let $p'$ and $q'$ be the parents of $p$ and $q$, respectively.
	Let $p''$ and $q''$ be the parents of $p'$ and $q'$, respectively.
	Let $\lambda=\lambda(p,q)$, $\lambda'=\lambda(p',q')$ and $\lambda''=\lambda(p'',q'').$
	We make a proof by induction on $\lambda$.
	
	We first handle the basis cases.
	\begin{itemize}
		\item If $\lambda=0$, then $p=q$ and thus $d_2(p,q)=d_1(p,q)=0$.
		\item If $\lambda=1$, then $p$ and $q$ are horizontal neighbors, and thus $d_2(p,q)=d_1(p,q)=1$.
		\item If $\lambda=2$, then $\lambda'=1$ by Lemma~\ref{p:lambda}.
			It follows that $d_2(p,q)=3$ and $d_1(p,q)=2$.
		\item If $\lambda=3$, then $1 \leq \lambda' \leq 2$, and $\lambda'' \leq 1$ 
			by Lemma~\ref{p:lambda}.
			It follows that $d_1(p,q)=3$ and $d_2(p,q) \leq 5$.
		\item If $\lambda=4$, then $\lambda'=2$ and $\lambda''=1$. It follows that
			$d_1(p,q)=4$ and $d_2(p,q)=5$.
	\end{itemize}
	
	Now suppose that $\lambda \geq 5$. It follows from Lemma~\ref{p:lambda} that 
	$3 \leq \lambda' \leq (1+\lambda)/2 < \lambda$. So by the induction hypothesis,
	we have $d_2(p',q')\leq d_1(p',q')+2$. As $\lambda \geq 2$, we have
	$d_2(p,q)=d_2(p',q')+2$, and thus 	$d_2(p,q) \leq d_1(p',q')+4$.
	By Lemma~\ref{lem:d1samerow}, we also
	have $d_1(p,q)=2+d_1(p',q')$. It follows that
	$
	  d_2(p,q) \leq d_1(p,q)+2.
	$

\end{proof}

We can now prove the main result of this section.

\begin{theorem}\label{th:discmain}
	For any two points $p,q \in \discrete^2$, we have
	$
		d_1(p,q) \leq d_2(p,q) \leq d_1(p,q)+2. 
	$
\end{theorem}
\begin{proof}
	We have $d_1(p,q) \leq d_2(p,q)$ because $d_1(p,q)$ is the length of a shortest path
	from $p$ to $q$ in $\discrete^D$, and $d_2(p,q)$ is the length of some path from $p$ to $q$.
	We now prove the other side of the inequality. 

	If $p$ and $q$ are at the same level, then the result is given by Lemma~\ref{lem:ineqsamerow}.
	So we may assume that $p$ and $q$ are not at the same level. Without loss of generality,
	we assume that the $\level(p) \leq \level(q)$. Let $r$ be the ancestor of
	$p$ that is at the same level as $q$, and let $i=\level(r)-\level(p)$.
	
	By definition, we have $d_2(p,q)=i+d_2(r,q)$. By Lemma~\ref{lem:d1shortest}, there is
	a shortest path from $p$ to $q$ going through $r$, and hence
	$d_1(p,q)=i+d_1(r,q)$. By Lemma~\ref{lem:ineqsamerow}, we have $d_2(r,q) \leq d_1(r,q)+2$.
	It follows that $d_2(p,q) \leq d_1(p,q)+2$.
\end{proof}

The bound in Theorem~\ref{th:discmain} is tight: In \figurename~\ref{fig:moves}b, 
we have $d_2(p,q)=d_1(p,q)+2=5$. Finally, we give a property of $d_2$-paths that will be needed
later.

Let $p,q \in \discrete^D$ be such that $p$ is neither a parent nor a descendant of $q$, and
$p\neq q$. Then the $d_2$-path $\rho$ from $p$ to $q$ contains exactly one horizontal move,
from a point $\bar p$ to $\bar q$. We call the edge $\bar p \bar q$ the {\it bridge} of $\rho$.
For instance, in \figurename~\ref{fig:moves}b, the bridge is the top edge of the blue path.
The {\it level} of the bridge is $\level(\bar p)=\level(\bar q)$, which we denote $\level(p,q)$.

When $p=q$, or $p$ is a descendant of $q$, we let $\level(p,q)=\level(q,p)=\level(q)$.
The lemma below allows us to approximate the level of a bridge.

\begin{lemma}\label{lem:bridge}
	Let $p$ and $q$ be two distinct points in $\discrete^D$ such that $p$ is neither a parent 
	nor a descendant of $q$. 	Let $\ell=\lfloor \log_2  \|x(p)x(q)\|_{\infty} \rfloor$.
	Then $\level(p,q)-1 \leq  \ell \leq \level(p,q)$.
\end{lemma}
\begin{proof}
	Without loss of generality, we assume that $\level(p) \leq \level(q)$.
	Let $i=\level(p,q)$, and let $\bar p \bar q$ be the bridge of the $d_2$-path from $p$ to $q$.
	As $\cell(\bar p)$ and $\cell(\bar q)$ are horizontal neighbors at level $i$, and $x(p)$ and $x(q)$
	lie in the vertical projection of these cells onto the hyperplane $z=0$, we have
	$\|x(p)x(q)\|_\infty < 2^{i+1}$, and thus $\ell \leq i$.
	
	Suppose that $\level(q)=i$, and hence $q=\bar q$. As $x(p)$ is not in the projection
	of $\cell(q)$ onto the hyperplane $z=0$, we have  $\|x(p)x(q)\|_\infty \geq 2^{i-1}$.
	It follows that $\ell \geq i-1$.
	
	The remaining case is when $\level(p) \leq \level(q) \leq i-1$. 
	Let $p'$ be the first
	point after the bridge vertex $\bar p$ on the $d_2$-path from $\bar p$ to $p$, 
	hence $p'$ is a child
	of $\bar p$. We define $q'$ in the same way. The points $x(p)$ and $x(q)$  are in the 
	vertical projection of the cells $\cell(p')$ and $\cell(q')$ onto the hyperplane $z=0$.
	As $p'$ and $q'$ are at level $i-1$ and are not neighbors, it follows that
	$\|x(p)x(q)\|_\infty \geq 2^{i-1}$, and thus $\ell \geq i-1$.
\end{proof}


\section{Embedding $\poincare^D$ into the discrete models} \label{sec:embedding}

In this section, we give an embedding of the 
Poincar\'e half-space $\poincare^D$ into the discrete models $(\discrete^D,d_1)$ and $(\discrete^D,d_2)$ with additive distortion $O(\log D)$. 
Our embedding maps any point $p \in \poincare^D$ to the center $b(\cell)$ of the cell $\cell$ of
$\poincare^D$ that contains $p$.
If several cells contain $p$, we take $b(p)$ to be the center of a cell at the maximum level.

We first observe that for all $x\geq 1$, 
	\begin{equation}\label{eq:UBarsin}
			\arsinh(x)=\ln(x+\sqrt{x^2+1})\leq \ln(x+\sqrt{2x^2}) \leq \ln(x)+\ln(1+\sqrt{2}) < \ln(x)+1.
	\end{equation}
We can then bound the hyperbolic distance between $p$ and $b(p)$.

\begin{lemma}\label{lem:embedding1}
For any $p \in \poincare^D$, we have $d_{H} (p,b(p)) < \ln D$.
\end{lemma}
\begin{proof}
	Let $i=\level(p)$.
	We have $\|pb(p)\| \leq 2^{i-1}\sqrt{D}$, $z(p) \geq 2^i$ and $z(b(p)) \geq 2^i$. Therefore
	\begin{align*}
		d_H(p,b(p)) & =2\arsinh\left(\frac 1 2 \frac{\|pb(p)\|}{ \sqrt{z(p)z(b(p))} }\right) \\
		& \leq 2\arsinh\left(\frac 1 2 \frac{2^{i-1}\sqrt D}{ 2^i }\right)=2\arsinh\left(\frac{\sqrt{D}}{4}\right)
	\end{align*}
	so by Inequality~\eqref{eq:UBarsin},
	\[
		d_H(p,b(p)) \leq 2 \ln(\sqrt D/4)+2=\ln(D)+2-4\ln 2 < \ln D.
	\]
\end{proof}

Let $p,q \in \poincare^D$ and let $h$ be the highest point on the geodesic from $p$ to $q$. 
Let $p'$ and $q'$ be the points vertically above $p$  and $q$, respectively, 
that are on the horizontal line through $h$.
(See \figurename~\ref{fig:semicircle}a and ~\ref{fig:semicircle}b.)

\begin{figure}
	\centering
	\includegraphics[width=\textwidth]{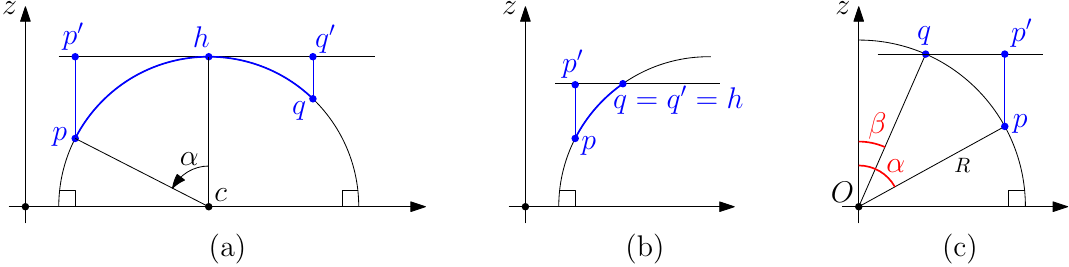}
	\caption{Lemma~\ref{lem:semicircle}\label{fig:semicircle}}
\end{figure}

\begin{lemma}\label{lem:semicircle}
For any $p,q \in \poincare^D$, we have
$0 \leq d_H(p,q)-d_H(p,p')-d_H(q,q') < 2$.
\end{lemma}

\begin{proof}
We first consider the special case where the geodesic from $p$ to $q$ is an arc of a circle
of radius $R$ 
centered at the origin $O$, and the geodesic from $p$ to $q$ does not cross the $z$-axis. Let $\alpha$ 
and $\beta$ be the 
angles that $Op$ and $Oq$ make with the $z$-axis, respectively.
Without loss of generality, we assume that  $\alpha>\beta$. (See \figurename~\ref{fig:semicircle}c.) 

By the triangle inequality, we have
\begin{align*}
d_H(p,q)-d_H(p,p') & \leq d_H(p',q) 
= 2\arsinh\left(\frac 12 \sqrt{\frac{R^2(\sin \alpha-\sin \beta)^2} 
{R^2\cos^2\beta}}\right) \\
& \leq 2\arsinh\left(\frac 12 \sqrt{\frac{(1-\sin \beta)^2} 
{\cos^2\beta}}\right)   \\
&  \leq  2\arsinh\left(\frac 12 \sqrt{\frac{1-\sin^2 \beta} 
{\cos^2\beta}}\right) \\
& = 2\arsinh(1/2)  < 1. \\
\end{align*}

The inequality above applies whenever $q$ is the highest point on the
geodesic arc $pq$, as in \figurename~\ref{fig:semicircle}b.
When the geodesic arc $pq$ has a local maximum $h$ in its interior
(See \figurename~\ref{fig:semicircle}a), 
we apply the inequality above twice through the point $h$,
 which yields the result. (See \figurename~\ref{fig:semicircle}a.)
\end{proof}

Still using the same notation as in \figurename~\ref{fig:semicircle}, 
we have the following two observations.
\begin{lemma}\label{lem:embedding0}
	Let $\hat p=b(p')$ and $\hat q=b(q')$. Then $d_1(\hat p,\hat q) \leq 5$.
\end{lemma}
\begin{proof}
	Let $R$ be the radius of the circle containing the geodesic arc $pq$.
	Suppose that $p'q'$ is tangent to this geodesic arc at $h$. (See \figurename~\ref{fig:semicircle}a.)
	Then $z(h)=R$, and the cells intersected by $p'q'$ are at a level $i$ such that 
	$2^i \leq R < 2^{i+1}$. So we have $i=\level(\hat p)=\level(\hat q)=\lfloor \log_2 R \rfloor$.
	Since $\|p'q'\| \leq 2R$, we have $\|x(\hat p) x(\hat q)\|_\infty \leq 2R+2^i$, 
	and since the width of the cells 	at level $i$ is at least $R/2$,  the horizontal
	distance $\lambda(\hat p,\hat q)$ is at most $5$. It follows that $d_1(\hat p, \hat q) \leq 5$.
	
	Now suppose that $h=q$ and the center of the circle containing the arc $pq$ is at the origin $O$. 
	(See \figurename~\ref{fig:semicircle}c.) 
	Let $\alpha$ and $\beta$ be the angles that $Op$ and $Oq$, respectively, make with the $z$-axis.
	Then the Euclidean length of 
	$p'q'=p'q$ is $R(\sin \alpha-\sin \beta)$ and the cells intersected by $p'q'$ have width
	at least $R \cos(\beta)/2$. Since
	\[ \frac{R(\sin \alpha-\sin \beta)}{R(\cos \beta)/2}
		\leq 2\frac{1-\sin \beta}{\cos \beta} 
		\leq 2\frac{1-\sin^2 \beta}{\cos^2 \beta}=2, \]
	by the same argument as above, $d_1(\hat p, \hat q) \leq 3$.
\end{proof}

The lemma below gives us a lower bound on the level of the bridge $\bar p \bar q$ when
$p'q'$ is tangent to the geodesic arc $pq$, as in \figurename~\ref{fig:semicircle}a.

\begin{lemma}\label{lem:bridge_circle}
	Suppose that $p$ and $q$ are two distinct points in $\discrete^D$ such that $p$ is neither 
	a parent nor a descendant of $q$, and 
	$h$ is between $p$ and $q$ along the geodesic arc $pq$. 
	Then we have $\level(p,q) \geq \log_2(z(h))-\log_2(D)/2-1$.
\end{lemma}
\begin{proof}
	The geodesic semi-circle through $p$ and $q$ has radius $z(h)$ and center $c=(x(h),0)$.
	Let $\alpha$ be the angle between the segment $cp$ and the segment $ch$.
	
	If $\alpha \leq \pi/3$, then $z(p) \geq z(h)/2$. Since
	$\level(p,q) \geq \level(p)$ and $z(p)=(3/2)2^{\level(p)}$, it follows that
	\[
		\level(p,q) \geq \level(p)= \log_2(z(p))-\log_2(3)+1 \geq \log_2(z(h))-\log_2(3).
	\]
	
	If $\alpha \geq \pi/3$, then $\|x(p) x(q)\| \geq \sqrt 3 z(h)/2$, so
	$$\|x(p)x(q)\|_\infty \geq \sqrt{3}z(h)/(2\sqrt{D}),$$ and thus by Lemma~\ref{lem:bridge},
	$$\level(p,q) \geq \log_2(\sqrt{3}z(h)/(2\sqrt{D})) \geq \log_2(z(h))-\log_2(D)/2-1.$$
\end{proof}

We now prove that $d_H$, $d_1$ and $d_2$ differ by at most $O(\log D)$ over $\discrete^D$.
We begin with two special cases, and then we handle the general case.

\begin{lemma}\label{lem:ancestor_circle}
	For any $p,q \in \discrete^D$ such that $q$ is an ancestor of $p$, 
	we have $$\ln(2) \cdot d_1(p,q) \leq d_H(p,q) \leq \ln(2) \cdot d_1(p,q)+\ln(D)+2+\ln 4.$$
\end{lemma}
\begin{proof}
	Let $i=\level(p)$ and $j=\level(q)$. Let $r=(x(q),z(p))$.
	Then we have
	\[
		d_H(p,q) \geq d_H(r,q)=\ln(z(q)/z(p))=\ln(2^{j-i})=(j-i)\ln 2.
	\]
	and thus
	\begin{equation}\label{eq:lem:ancestor_circle}
		d_H(p,q) \geq (j-i) \ln 2.
	\end{equation}
	We now prove an upper bound on $d_H(p,q)$. Observe that $\|pq\| \leq 2^{j+1}\sqrt{D}$.
	\begin{align*}
		d_H(p,q)  & = 2\arsinh\left(\frac 1 2 \frac{\|pq\|}{\sqrt{z(p)z(q)}} \right) \\
					 & \leq 2\arsinh\left(\frac{2^{j+1}\sqrt D}{3\cdot 2^{(i+j)/2}} \right)
					 \leq 2\arsinh\left({2^{(j-i+2)/2}\sqrt D} \right)
	\end{align*}
	Then it follows from Inequality~\ref{eq:UBarsin} that
	\begin{align*}
		d_H(p,q) & \leq 2 \ln\left(2^{(j-i+2)/2}\sqrt D\right)+2 \\
		&=\ln(2^{j-i+2})+\ln(D)+2=(j-i+2) \ln(2) + \ln(D)+2
	\end{align*}
	and thus $d_H(p,q) \leq (j-i)\ln(2)+\ln(D)+2+\ln 4$.
	
	Together with \eqref{eq:lem:ancestor_circle}, it shows that 
	$$(j-i)\ln 2 \leq d_H(p,q) \leq (j-i) \ln(2)+\ln(D)+2+\ln 4.$$
	The result follows from $d_1(p,q)=j-i$.
\end{proof}

\begin{figure}
	\centering
	\includegraphics{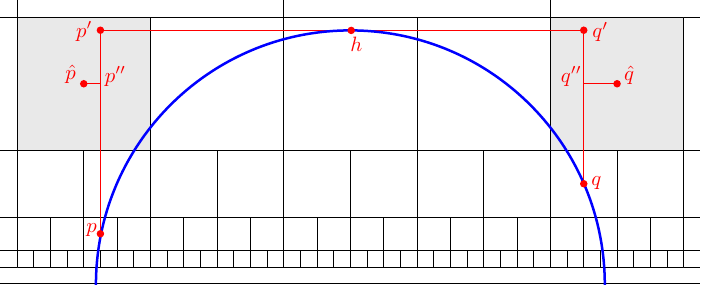}
	\caption{Proof of Lemma~\ref{lem:circle}\label{fig:circle}}
\end{figure}

\begin{lemma}\label{lem:circle}
	For any $p,q \in \discrete^D$, we have 
	\[
	 -7 \ln 2 < d_{H} (p,q)-\ln(2)\cdot d_1(p,q) \leq \ln(D)+2+6 \ln 2.
	\]
\end{lemma}
\begin{proof}
	Without loss of generality, we assume that $\level(p) \leq \level(q)$. 
	The case where $q$ is an ancestor of $p$ is handled by Lemma~\ref{lem:ancestor_circle},
	so we assume that $q$ is not an ancestor of $p$.
	Let $p',q'$ be defined as above. (See \figurename~\ref{fig:circle}.) Let $\hat p=b(p')$, 
	$\hat q=b(q')$, $p''=(x(p),\hat z(p))$ and $q''=(x(q),\hat z(q))$.	
		Let $i=\level(p)$ and $j=\level(\hat p)=\level(\hat q)$.
	
	By Lemma~\ref{lem:embedding0}, 	we have $d_1(\hat p,\hat q) \leq 5$, and thus
	\begin{equation} \label{eq:lem:circle:1}
		d_1(p,q)  \leq d_1(p,\hat p)+d_1(q,\hat q)+5.
	\end{equation}
	
	Suppose that $q=\hat q$, and thus $d_1(q,\hat q)=0$. 
	Then the $d_1$-path from $p$ to $q$ goes through $\hat p$,
	and we have $d_1(p,q) \geq d_1(p,\hat p)$, therefore $d_1(p,q) \geq d_1(p,\hat p)+d_1(q,\hat q)$.
	
	Now suppose that $q \neq \hat q$. It implies that $p \neq \hat p$, and that
	$h$ is between $p$ and $q$ along the geodesic arc $pq$. So by 
	Lemma~\ref{lem:bridge_circle}, we have $\level(p,q) \geq \log_2 z(h)-\frac 12\log_2 D-1$.
	As $z(h) \geq 2^{j}$, it implies that $\level(p,q) \geq j- 1- \frac 12\log_2 D$. Hence the $d_2$-path
	$pq$ follows the path $p \hat p$ until $\hat p$, or a descendant of $\hat p$
	at most $1+\frac 12\log_2 D$ levels below $\hat p$.
	So the portion of the $d_2$-path $pq$ before the bridge has length at least
	$d_1(p,\hat p)-\frac 12\log_2 D-1$. Similarly, the portion after the bridge has
	length at least $d_1(q,\hat q)-\frac 12\log_2 D-1$.
	Therefore, $d_2(p,q) \geq d_1(p,\hat p)+d_1(q,\hat q)-2- \log_2 D$.
	By Theorem~\ref{th:discmain}, it implies that 
	$d_1(p,q) \geq d_1(p,\hat p)+d_1(q,\hat q)-4- \log_2 D$.
	
	Together with Equation~\eqref{eq:lem:circle:1}, it implies that, whether $q =\hat q$ or not,
	\[
		 -5 \leq d_1(p,\hat p)+d_1(q,\hat q) - d_1(p,q) \leq 4+\log_2 D.
	\]
	We also have $d_1(p,\hat p)=j-i$ and thus
	\begin{align*}
		d_H(p,p'')=\ln(z(p'')/z(p)) & =
		\ln\left(\left(3\cdot 2^{j-1}\right)/\left(3\cdot 2^{i-1}\right)\right)\\
		& =(j-i)\ln 2 =\ln(2) \cdot d_1(p,\hat p).
	\end{align*}
	Similarly, we have $d_H(q,q'')=\ln(2)\cdot d_1(q,\hat q)$, and thus
	\begin{align*}
		-5 \ln 2 & \leq d_H(p,p'')+d_H(q,q'')- \ln(2)\cdot d_1(p,q) \\ & 
		\leq \ln(2) \cdot \log_2(D)+4 \ln 2 
		= \ln(D)+4 \ln 2.
	\end{align*}
	We also have
	$d_H(p',p'')= |\ln(z(p')/z(p'')| \leq \ln 2$ and, 
	similarly, $d_H(q',q'') \leq \ln 2$.
	It follows that 
	\[
		-7 \ln 2 \leq d_H(p,p')+d_H(q,q') - \ln(2)\cdot d_1(p,q) \leq \ln(D)+6 \ln 2.
	\]
	The result  follows from Lemma~\ref{lem:semicircle}.
\end{proof}

It follows from Theorem~\ref{th:discmain}, 
Lemma~\ref{lem:embedding1} and Lemma~\ref{lem:circle} that our
embedding $b:\poincare^D \to \discrete^D$ gives $O(\log D)$ additive distortion with
respect to $d_1$ and $d_2$.
\begin{theorem} \label{th:embedmain}
	For any $p,q \in \poincare^D$, we have
	$|d_H(p,q)-\ln(2) \cdot d_1(b(p),b(q))| \leq 3 \ln(D)+O(1)$ and 
	$|d_H(p,q)-\ln(2) \cdot d_2(b(p),b(q))| \leq 3 \ln(D)+O(1)$. 
\end{theorem}


\section{Compressed quadtrees} \label{sec:quadtree}

In this section, we present {\it compressed quadtrees}, 
which will be needed for our graph metric embeddings and for approximate near neighbor searching. 
Compressed quadtrees are presented in the book by Har-Peled~\cite{HPbook}.
Our quadtrees will record points in $\discrete^{D}$, which correspond to quadtree boxes
in $\IR^{D-1}$. (See the discussion below.) Hence, our quadtrees record a set of quadtree
boxes instead of recording a set of points in $\IR^{D-1}$. This does not affect the
time and space bounds~\cite[Lemma 2.11]{HPbook}.

Let $\cell$ be a cell of our binary tiling of $\poincare^D$. 
The {\it vertical projection} $\cell_x$ of $\cell$ is the set of the projections $x(p)$ of
all the points $p \in \cell$. Thus, $\cell_x$ is a box in $\IR^{D-1}$, and more precisely,
$\cell_x$ is of the form 
\[
	[k_12^{-i},(k_1+1)2^{-i}] \times \dots \times [k_{D-1}2^{-i},(k_{D-1}+1)2^{-i}]
	\text{ where } k_1,\dots,k_{D-1} \in \IZ.
\]
A box $\cell_x$ of this form is called a {\it quadtree box}.
Thus, each point $b \in \discrete^D$ corresponds to exactly one quadtree box
$\cell_x(b) \subset \IR^{D-1}$, and the corresponding cell $\cell(b)$ of our binary tiling of $\poincare^D$.

It allows us extend the notions of children, parents, ancestors, descendants and levels to
quadtree boxes.
So for any two cells $\cell$ and $\cell'$ of our binary tiling, we say that $\cell_x$
is a child (resp. parent, ancestor, descendant) of $\cell'_x$ if $\cell$ is
a child (resp. parent, ancestor, descendant) of $\cell'$. The level $\level(\cell_x)$
of the quadtree cell $\cell_x$ is $\level(\cell)$.

Let $Q$ be a set of $n$ points in $\discrete^D$, such that $x(Q) \subset (0,1)^{D-1}$.
A compressed quadtree $\mathcal T$ storing $Q$ is a tree such that each node $\nu$
of $\mathcal T$ records one cell $\cell(\nu)$ of our binary tiling, and, equivalently,
the quadtree box $\cell_x(\nu)$ and the center point $b(\nu)=b(\cell) \in \discrete^D$. 
This tree is constructed as follows. 
A cell $\cell$  that contains no point of $Q$, and 
such that no descendant of $\cell$ contains a point of $Q$, is said to be {\it empty}. A
cell $\cell(\nu)$ where $\nu$ is a node of $\mathcal T$ is a {\it leaf cell} if it contains at most one point of $P$.
If $\cell(\nu)$ is a leaf cell, then $\nu$ is a leaf node of $\mathcal T$.
The root records the cell $[0,1]^{D-1}\times [1,2]$. The children of a node $\nu$
that is not a leaf are constructed as follows:
\begin{itemize}
	\item If two or more children cells of $\cell(\nu)$ are non-empty, then we create a node
		$\nu_i$ corresponding to each such child cell, and make $\nu_i$ a child
		of $\nu$. In this case, $\nu$ is an {\it ordinary node}.
		(See \figurename~\ref{fig:Qnodes}a.)
	\item If only one child of $\cell(\nu)$ is non-empty, 
		let $\nu_1$ be the descendant of $\nu$ such that $\cell_x(\nu_1) \cap Q =\cell_x(\nu) \cap Q$
		and the level of $\nu_1$ is minimum. Then $\nu_1$ is the only child of $\nu$, and
		$\nu$ is a {\it compressed node}. The {\it region} associated with
		$\nu$ is $\rg(\nu)=\cell_x(\nu)-\cell_x(\nu_1)$.
		(See \figurename~\ref{fig:Qnodes}b.)
\end{itemize} 

\begin{figure}
	\begin{center}	
		\includegraphics[width=\textwidth]{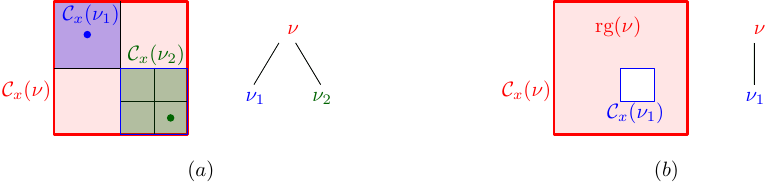}
	\end{center}
	\caption{Nodes of a compressed quadtree. (a) An ordinary node $\nu$, with two children cells $\nu_1$ and $\nu_2$.
	(b) A compressed node $\nu$, with its child $\nu_1$. The region $\rg(\nu)$ is
		shaded in red.\label{fig:Qnodes}}
\end{figure}

The quadtree boxes stored at the leaf nodes and the regions of the compressed nodes form a partition of $[0,1]^{D-1}$.
This compressed quadtree $\mathcal T$ has $O(n)$ nodes, it can be constructed in $O(Dn \log n)$ time,  
and the node corresponding to the leaf cell or compressed node region containing a query point $q \in \IR^{D-1}$ can be 
found in time $O(D \log n)$ time~\cite[Chapter 2]{HPbook}.

In addition, we can answer {\it cell queries} in $O(D \log n)$ time: Given a query quadtree box
$\cell_x$, we can find the largest box $\cell_x^-$ stored in $\mathcal T'$ such that 
$\cell_x^- \subset \cell_x$, and we can find the smallest box $\cell_x^+$ stored in $\mathcal T'$ that
contains $\cell_x$.


\section{Spanner in the  discrete model}\label{sec:spannerdiscrete}

In this section, we give a 2-additive spanner $S$ in the first discrete model $(\discrete^D,d_1)$,
for a set $P\subset \discrete^D$ of $n$ points.
Our approach is the following: We overlay the $d_2$-paths of all pairs of points
in $P$, and add a new vertex, called a {\it Steiner} vertex, at each point where such
a path bends---in other words, we add a Steiner vertex at each endpoint of each bridge,
if this endpoint is not in $P$. (See~\figurename~\ref{fig:size}.) We also add a Steiner vertex
whenever two such $d_2$-paths merge. The length of an edge is the number of cell boundaries that it crosses.

\begin{figure}
\centering
	\includegraphics{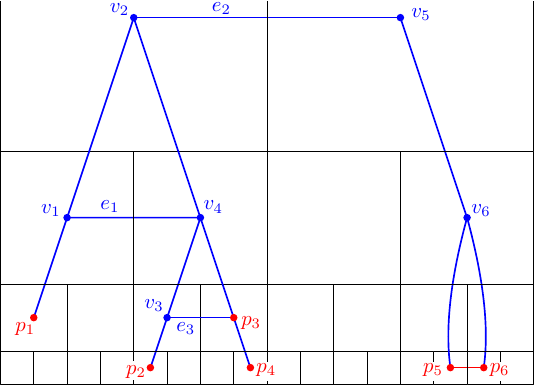}
	\caption{Our construction of a 2-additive spanner $S$ for the points $p_1,\dots,p_6$
		with respect to the metric $d_1$. 
		The Steiner vertices are $v_1,\dots,v_6$. The vertices $v_1$ and $v_4$ are the bending
		points of the $d_2$-path from $p_1$ to $p_2$, hence $v_1v_4$ is the bridge between $p_1$
		and $p_2$. The vertex $v_3$ is the bending point of
		the $d_2$-path from $p_2$ to $p_3$. The distance from $p_3$ to $p_6$ within $S$ is
		$d_S(p_3,p_6)=6$, as the edge $v_6p_6$ has length 2. On the other hand, we have $d_1(p_3,p_5)=4$.
		The vertex $v_6$ is the Steiner vertex at which the $d_2$-paths from $p_5$ and $p_6$ to $p_3$
		merge.
	\label{fig:size}}
\end{figure}

The resulting graph $S$ contains, by construction, a path
of length $d_2(p,q)$ between any two points $p,q \in P$. By Theorem~\ref{th:discmain}, it
follows that:
\begin{proposition}\label{prop:spannererror}
	For any $p,q \in P$, there is a path from $p$ to $q$  in the graph $S$
	of length at least $d_1(p,q)$ and at most $d_1(p,q)+2$.
\end{proposition}

Using compressed quadtrees (Section~\ref{sec:quadtree}), we now give a bound on
the size of $S$, as well its construction time.
Let $P$ be our input set of points in $\poincare^D$, and let $x(P)=\{x(p) : p \in P\}$
be its vertical projection. Without loss of generality, we assume that 
$x(P) \subset [0,1]^{D-1}$. If it were not the case, we could apply to $P$ a scaling 
transformation centered at $O$ followed by a horizontal translation so that 
the vertical projection of $P$ is in $[0,1]^{D-1}$. The hyperbolic distance $d_H(\cdot)$
is invariant under these transformations.
We construct the compressed quadtree $\mathcal T$ that records the points in $P$.

Let $\nu$ be a node in $\mathcal T$, and let $r$ be the point in $\discrete^D$ that
corresponds to $\cell(\nu)$, so  $\cell_x(r)=\cell_x(\nu)$. Given a horizontal neighbor
$r'$ of $r$, we can check whether $rr'$ is a bridge as follows. First check using
cell queries in $\mathcal T$ whether the cells $\cell(r)$ and $\cell(r')$
are non-empty. If so, check whether, among the children of $\cell(r)$ and
$\cell(r')$, there is a pair of non-empty cells that are not horizontal neighbors.
For a given horizontal neighbor $r'$ or $r$, we can check
it using $2^{O(D)}$ cell queries, and since there are $2^{O(D)}$ horizontal
neighbors $r'$, and we can do it using $2^{O(D)}$ queries in total.
We perform this operation for each node $\nu$ of $\mathcal T$.
As there are $O(n)$ nodes in $\mathcal T$, and cell queries can be answered in $O(D \log n)$
time, this process generates $2^{O(D)}n$ bridges in $2^{O(D)}n\log n$ time.

We still need to find bridges that do not correspond to any node $\nu \in \mathcal T$,
that is, bridges $rr'$ such that neither $\cell(r)$ nor $\cell(r')$ is recorded
in $\mathcal T$.
As the quadtree boxes stored at leaf cells and the regions associated with the compressed nodes form a partition
of $[0,1]^{D-1}$, such a bridge $rr'$ must satisfy $x(r) \in \rg(\nu)$ and
$x(r') \in \rg(\nu')$ for some compressed nodes $\nu,\nu'$. Then $rr'$ must connect
the points of $\discrete^D$ corresponding to the cell of the children of $\nu$ and $\nu'$.
Given $\nu$ and $\nu'$, we  can determine the
level of this bridge, if it exists, in $O(1)$ time, using Lemma~\ref{lem:bridge}. 
We check the existence of such a bridge
for all pairs of compressed nodes $\nu,\nu'$ whose cells are adjacent. 
There are $2^{O(D)}n$ such pairs to check, because any cell $\cell_x(\nu)$ is adjacent
to $2^{O(D)}$ cells of size at least the size of $\cell_x(\nu)$.
So again, this process generates $2^{O(D)}n$ bridges.

The construction above also gives us the vertical edges descending from every
Steiner point. It follows that:

\begin{theorem}\label{th:spanner_d1}
	Given a set $P \subset \discrete^D$ of $n$ points, we can compute in time $2^{O(D)}n \log n$ 
	a $2$-additive spanner $S$ of $P$
	that has $2^{O(D)}n$ Steiner vertices and edges. More precisely, $S$ is 
	a weighted graph embedded 	in $\discrete^D$,
	any of its edge $uv$ has length $d_1(u,v)$, and for any $p,q \in P$, there
	is a path in $S$ from $p$ to $q$ that has length $d_2(p,q)$, where $d_1(p,q) \leq d_2(p,q) \leq d_1(p,q)+2$.
\end{theorem}


\section{Embedding into a graph metric and spanner for $\poincare^D$}\label{sec:hspanner}

We can embed any set $P$ of $n$ points in $\poincare^D$ into a graph metric as follows.
First we map each $p \in P$ to the point $b(p) \in \discrete^D$. Then we construct the
spanner $S$ for these points with respect to $d_1$ as in Theorem~\ref{th:spanner_d1}. 
We multiply all the edge lengths of $S$ by $\ln(2)$, thus obtaining a graph $S'$ 
such that the length $d_{S'}(p,q)$ of a shortest path in $S'$ between any two points 
$p,q \in \discrete^D$ satisfies
\[
	\ln(2)\cdot d_1(p,q) \leq d_{S'}(p,q) \leq \ln(2)\cdot d_2(p,q).
\]

By Theorem~\ref{th:embedmain}, this path length approximates $d_H(p,q)$ within an additive
error  $O(\log D)$. In summary:

\begin{corollary}\label{cor:embed}
	Given a set $P \subset \poincare^D$ of $n$ points, we can compute in time $2^{O(D)}n \log n$ a 
	positively weighted graph $G(V,E)$ 
	that has $2^{O(D)}n$ vertices and edges, and a mapping $b:P \to V$, such that
	$|d_H(p_i,p_j)-d_G(b(p_i),b(p_j))|=O(\log D)$ for any $p_i,p_j \in P$.
\end{corollary}

The graph from Corollary~\ref{cor:embed}
is not a spanner in the sense that the length of an edge $b(p)b(q)$, $p,q \in \poincare^D$ is 
$\ln(2) \cdot d_1(b(p),b(q))$ instead 
of $d_H(p,q)$. If we set the edge weights to be $d_H(b(p),b(q))$ for each edge $b(p)b(q)$, then we introduce
a constant additive error for each edge of $S$, and since a shortest path in $S$ may consist of $\Omega(n)$
edges, we will no longer have $O(\log D)$ additive error.

In order to obtain an additive spanner 
for $P \subset \poincare^D$ that is embedded in $\poincare^D$, we will add more edges, which will act as shortcuts.
So let $T$ be the forest obtained by removing the horizontal edges from $S$, and let $k$ be a fixed integer.
We orient all the edges of $T$ upwards.
We construct a $k$-transitive closure spanner $T_k$ of $S$~\cite{Raskhodnikova10,Thorup97}. This graph $T_k$ has $O(n \cdot \lambda_k(n))$ edges, where $\lambda_k(n)$ is the $k$-th row of the 
inverse Ackermann function. The vertices of $T_k$ are the points $b(p)$, $p \in P$ and
the edges of $T_k$ are a superset of the edges of $T$, also oriented upwards. The key property of $T_k$ 
is that, if there is a path from vertex $p$ to $q$ in $T_k$, then there exists such a path from $p$ to $q$ of
length at most $k$. The $k$-transitive closure spanner can be computed in $O(n \cdot \lambda_k(n))$ 
time~\cite[Algorithm L]{Thorup97}.

Our Spanner $S_k$ is obtained from $T_k$ by adding all the horizontal edges of $S$, and by adding 
each point $p \in P$ as a vertex, together with the edge $pb(p)$. 
Each edge $pq$ of $S_k$ is assigned the weight $d_H(p,q)$. In particular, an edge $pb(p)$ has
weight $O(\log D)$ by Lemma~\ref{lem:embedding1}.

Let $p,q \in P$. We now prove that there is a path of length $d_H(p,q)+O(k \log D)$ in $S_k$. By construction, there
is a path in $S$ from $b(p)$ to $b(q)$ of length $d_2(b(p),b(q))$ that consists of a vertical path $\rho$ up from $b(p)$ 
to a vertex $\bar b(p)$, then a horizontal edge from $\bar b(p)$ to a vertex $\bar b(q) $, 
followed by a vertical path $\rho'$ down from $\bar b(q)$ to $b(q)$.

In $T_k$, the path $\rho$ can be replaced with a path $\rho_k$ of at most $k$ edges. 
The length of each edge $vw$ in $\rho_k$ is $d_H(v,w)$, which is $\ln(2) \cdot d_2(v,w)+O(\log D)$ by Theorem~\ref{th:embedmain}.
So the length of $\rho_k$ (i.e. the sum of the weights of its edges) is $\ln(2) \cdot d_2(b(p),\bar b(p))+O(k\log D)$. Similarly, 
there is a path $\rho'_k$ from $\bar b(q)$ to $b(q)$ in $T_k$ with length 
$\ln(2) \cdot d_2(\bar b(q),b(q))+O(k\log D)$. It follows that the 
length of the path in $S$ consisting of $pb(p)$ followed by $\rho_k$, $\bar b(p)\bar b(q)$, $\rho'_k$ and $qb(q)$ has 
length $\ln(2) \cdot d_2(b(p),b(q))+O(k\log D)$. By Lemma~\ref{lem:embedding1} and Theorem~\ref{th:embedmain}, this length 
is $d_H(p,q)+O(k\log D)$.
\begin{theorem}\label{th:spannerH2}
	Let $k$ be an integer, and let $P$ a set of $n$ points in $\poincare^D$. We can construct in $2^{O(D)}n\log n$
	time an $O(k\log D)$ purely additive spanner of $P$ with $2^{O(D)}n$ Steiner vertices and $2^{O(D)}n \cdot \lambda_k(n)$ edges.
\end{theorem}


\section{Approximate Voronoi diagram (AVD)}\label{sec:NN}

In this section, we first show how to answer exact near-neighbor queries 
in $(\discrete^D,d_2)$, 
and then we briefly explain how this result yields ANN data structures for 
$(\discrete^D,d_1)$ and $\poincare^D$ with additive error 2 and $O(\log D)$, respectively.

So let $P=\{p_1,\dots,p_n\}$ be a set of $n$ points in $\discrete^D$. 
For a query point $q \in \discrete^D$, we want to find a point $n_2(q) \in P$
such that $d_2(q,n_2(q))$ is minimum. As there can be several closest point to $q$,
we break the ties by taking $n_2(q)$ to be  the point $p_i$ such that 
$d_2(p_i,q)=\min_j d_2(p_j,q)$ and $i$ is minimum.

Without loss of generality, we assume that $x(P) \subset [1/4,1/2]^{D-1}$.
If $x(q) \notin [0,1]^{D-1}$, then $n_2(q)$ is the highest point in $P$, which
the data structure described below can return in constant time.

This data structure records an {\it approximate Voronoi diagram} (AVD), which is a partition
of $\discrete^D$ in $2^{O(D)}n$ regions $V_1,\dots,V_m$ such that each region $V_i$
is associated with a set $R_i \subset P$ of at most $2^{O(D)}$ {\it representative points}.
For any point $q \in V_i$, we have $n_2(q) \in R_i$.

In the definition above, the AVD allows us to find an {\it exact} nearest
neighbor with respect to the second discrete model $(\discrete^D,d_2)$. 
In the first discrete model $(\discrete^D,d_1)$
and in $H^{D}$, it will allow us to find an approximate near-neighbor, with
additive error 2 and $O(\log D)$, respectively.

As in Section~\ref{sec:spannerdiscrete}, we first compute the compressed quadtree
$\mathcal T$ that records our input point set $P$. This quadtree $\mathcal T$ has $O(n)$ nodes, it can
be computed in time $O(Dn \log  n)$ time, and it allows us to perform 
cell queries  in $O(D \log n)$ time. 

\paragraph{Refinement step}
Suppose that $\nu$ is an ordinary node of $\mathcal T$. Let $b(\nu)$ be the corresponding
point of $\discrete^D$, then $b(\nu)$ is the only point of $\discrete^D$ contained in $\cell(\nu)$.
So we associate $\nu$ with a Voronoi region consisting only of $b(\nu)$, and
with a single representative $n_2(b(\nu))$. 

On the other hand, if $\nu$ is a leaf node of $\mathcal T$, then the descendants of $\cell(\nu)$
do not correspond to nodes of $\mathcal T$. So we would like  $\nu$ to be associated with
a Voronoi region $V_i$ that covers $b(\nu)$ and all its descendants. The issue here is
that we may need $\Omega(n)$ representatives for $V_i$, if for instance, $\cell_x(\nu)$
is adjacent to $\Omega(n)$ smaller, disjoint cells of $\mathcal T$.

In order to solve this problem, we will break $\cell_x(\nu)$ into smaller cells, by
inserting more boxes into our quadtree. We call this step of the construction
the {\it refinement step}. For each
node $\nu$ of $\mathcal T$ such that $\cell(\nu)$ is not empty, 
we insert into our compressed quadtree all the boxes that are 
horizontal neighbors of $\cell(\nu)$,  thus obtaining  a quadtree $\mathcal T'$
that is associated with a finer subdivision of $\IR^{D-1}$. 
(See \figurename~\ref{fig:intuition}.) In particular, for a compressed node $\nu$,
we insert the neighbors of the outer cell and the inner cell, that is, the horizontal neighbors
 of $\cell(\nu)$ and $\cell(\nu_1)$, where $\nu_1$ is the child of $\nu$.
\begin{figure}	
	\includegraphics[width=\textwidth]{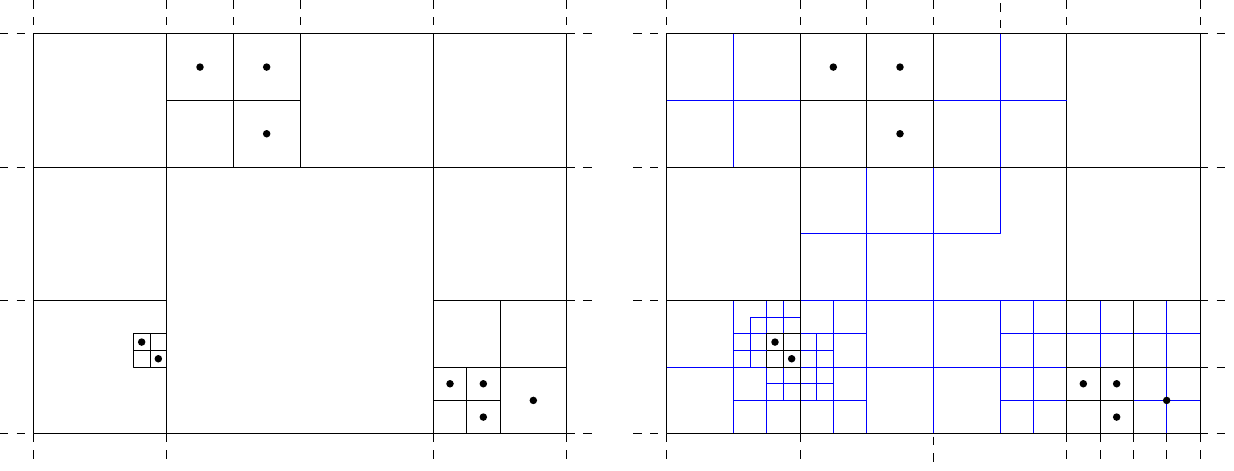}
	\caption{(Left) The subdivision corresponding to the initial quadtree $\mathcal T$.
	(Right) The subdivision after refinement, corresponding to the quadtree $\mathcal T'$.
	\label{fig:intuition}}
\end{figure}

We say that two quadtree boxes in $\IR^{D-1}$ are {\it adjacent} if they
intersect along their boundaries, but not in their interior.
The two lemmas below will help us bound the number of representative points for
each cell of our AVD. 

\begin{lemma}\label{lem:adj1}
	Let $\nu'$ be a leaf cell of $\mathcal  T'$. Then the number of 
 compressed nodes $\nu$ of $\mathcal T$ such that $\cell_x(\nu)$ 
	is adjacent to $\cell_x(\nu')$	is $2^{O(D)}$.
\end{lemma}
\begin{figure}
	\centering
	\includegraphics{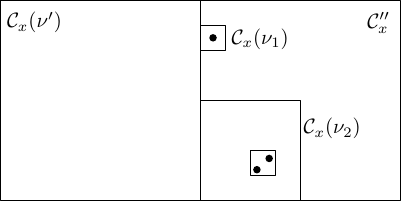}
	\caption{Proof of Lemma~\ref{lem:adj1}. The node $\nu_2$ is a compressed node.\label{fig:adj1}}
\end{figure}
\begin{proof}
	The leaf cells and the regions of the compressed nodes of $\mathcal T$ 
	form a subdivision of $\IR^{D-1}$.
	The leaf cells and the regions of the compressed nodes of $\mathcal T'$ 
	form a refinement of this subdivision.
	So there are at most $2^{O(D)}$ such
	cells $\cell_x(\nu)$ adjacent to $\cell_x(\nu')$ of size larger than or
	equal to the size of $\cell_x(\nu')$. So in the following, we consider 
	cells $\cell_x(\nu)$ whose size is smaller than the size of $\cell_x(\nu')$.

	For the sake of contradiction, suppose that there are at least $3^{D-1}$ such nodes
	$\nu$ such that the size of	$\cell_x(\nu)$ is smaller than that of $\cell_x(\nu')$.
	Then two of them, $\cell_x(\nu_1)$ and $\cell_x(\nu_2)$, must be contained
	in a quadtree box $\cell_x''$ that has the same size as $\cell_x(\nu)$ and
	is adjacent to it. (See \figurename~\ref{fig:adj1}.) This is impossible, as it
	would have caused $\cell_x(\nu')$ to be partitioned at the refinement step.		
\end{proof}
\begin{lemma}\label{lem:adj2}
	Let $\nu'$ be a compressed node of $\mathcal  T'$. Then the number of 
	compressed nodes	$\nu$ of $\mathcal T'$ such that $\cell_x(\nu)$ 
	is adjacent to $\rg(\nu')$	is $2^{O(D)}$.
\end{lemma}
\begin{proof}
	The same proof as for Lemma~\ref{lem:adj1} shows that there are $2^{O(D)}$ such 
	cells that are adjacent to $\rg(\nu')$ along its outer boundary. A similar argument
	again shows that there are $2^{O(D)}$ such cells adjacent along the inner boundary.
\end{proof}

\paragraph{The approximate Voronoi diagram}
Our AVD consists of one region $V(\nu')$ for each node $\nu'$ of $\mathcal T'$:
\begin{itemize}
	\item If $\nu'$ is an ordinary node, then the Voronoi region $V(\nu')$ is $\cell(\nu')$.
	\item If $\nu'$ is a leaf node, then the Voronoi region $V(\nu')$ is the union
		of $\cell(\nu')$ and all its descendants.
		In other words, $V(\nu')$ is the set of all the points that are
		in or below $\cell(\nu')$.
	\item If $\nu'$ is a compressed node, let $\nu'_1$ be its child. Then $V(\nu')$ is
		the  set of all the points that are on or below $\cell(\nu')$, and are not on or below
		$\cell(\nu'_1)$.
\end{itemize}

\paragraph{Preprocessing}

We first compute, for every node $\nu'$ of $\mathcal T'$ such that $\cell(\nu')$ is not empty, 
the point $h(\nu')$ in $P$
with highest $z$-coordinate that is recorded in the subtree rooted at $\nu'$.
(Again, we break ties by taking the point $p_i$ with smallest index $i$ among
these highest points.)
We can compute all these points in linear time by traversing the tree $\mathcal T'$ from bottom
to top.

Next, we compute the $d_2$-nearest neighbor $n_2(\nu')=n_2(b(\nu'))$ for all the nodes $\nu'$ of
$\mathcal T'$ as follows. We compute these points by traversing the tree from top to bottom, so
when we compute $n_2(\nu')$, we already know $n_2(\nu'')$ where $\nu''$
is the parent of $\nu'$. Let $\rho$ be the $d_2$-shortest path from $b(\nu')$ to
$n_2(\nu')$. 
The 4 possible cases are listed below. We compute the nearest point for each case, 
and record the nearest one as $n_2(\nu')$.
\begin{enumerate}
	\item If $\rho$ goes through $b(\nu'')$, then 
		$n_2(\nu')=n_2(\nu'')$, and we are done as we already know $n_2(\nu'')$.
	\item If $\rho$ is a downward path, or if $b(\nu) \in P$, then $n_2(\nu)=h(\nu)$.
	\item If $\rho$ starts with a horizontal move from $b(\nu')$, then 
		$n_2(\nu')$ is the highest point that is on or below a cell $\cell'$ that 
		is a horizontal neighbor of $\cell(\nu')$. 	We can find this highest point
		in $O(D \log n)$ time by performing a cell query in $\mathcal T$, and since there are $2^{O(D)}$
		such cells $\cell'$, we find $n_2(\nu)$ in $2^{O(D)} \log n$ time.
	\item Otherwise, $\rho$ first goes upward from $b(\nu')$, then follows a bridge $r'r$,
		and then goes downward to $n_2(\nu')$. As $\rho$ does not go through $b(\nu'')$, it
		follows that $\nu''$ is a compressed node. Then $\cell(r)$ is not stored 
		in $\mathcal T'$, because if it were the case, we would have created a node
		for $\cell(r')$ in $\mathcal T'$, and $\nu''$ could not be the parent of $\nu'$. 
		So $r$ is in $V(\nu)$ for some 
		compressed node $\nu$ of $\mathcal T$,  and we have $n_2(b(\nu))=h(\nu)$.  
		As $\cell_x(\nu')$ and $\cell_x(\nu)$ must be adjacent, by Lemma~\ref{lem:adj2},
		there are $2^{O(D)}$ candidates for $\nu$, so we can find it in $2^{O(D)}\log n$ time. 
\end{enumerate}

\paragraph{Finding the representative points}
We now explain how we choose the representative points for each region $V(\nu')$.
For each node $\nu'$ of $\mathcal T'$, we pick $n_2(\nu')$ as a representative point.
If $\nu'$ is an ordinary cell, then we do not add any other representative point.

If $\nu'$ is a leaf node or a compressed node, then for each compressed node $\nu$ of
$\mathcal T$ such that $\cell_x(\nu)$ is adjacent to $\cell_x(\nu')$ or $\rg(\nu')$, respectively, 
we add $h(\nu)$ 
as a representative point.  By Lemma~\ref{lem:adj1} and Lemma~\ref{lem:adj2}, there
are $2^{O(D)}$ such points.

We now prove that for each point $q$ in a Voronoi cell $V(\nu')$, 
 one of these representative points is
$n_2(q)$. So suppose that $n_2(q) \neq n_2(\nu')$, and thus $\nu'$ is
not an ordinary node. Let $\rho$ be the $d_2$-path
from $q$ to $n_2(q)$. As $\nu'$ is either a leaf node or a compressed node,
there is no point in $P$ below $q$, so $\rho$ cannot be a downward path.
As $n_2(q) \neq n_2(\nu')$, it cannot be an upward path either. So it must have
a bridge $r'r$, where $r$ is on or above $n_2(q)$. If $\cell(r)$ is recorded
in a node $\nu(r)$ of $\mathcal T$, then at the refinement stage, we
must have inserted into $\mathcal T'$ the node $\nu(r')$ such that
$\cell(\nu(r'))=\cell(r')$, and then $n_2(q)=n_2(\nu')$, a contradiction.
Therefore, $\cell(r)$ is not recorded in $\mathcal T$. Hence, $r$ is in
$\rg(\nu)$ for some compressed node $\nu$ of $\mathcal T$, such that
$\cell_x(\nu)$ is adjacent to $\cell_x(\nu')$. It follows that $n_2(q)=h(\nu)$,
which is one of our representative points. 

\paragraph{Result} As $\mathcal T'$ records $2^{O(D)}n$ nodes, it can be
computed in $2^{O(D)}n \log n$ time, and it can be queries in $O(D+\log n)$
time. Since our data structure returns an exact nearest neighbor with respect to $d_2$,
by Theorem~\ref{th:discmain}, it returns an approximate near neighbor with respect to
$d_1$, with additive error 2.
So we obtained the following result:
\begin{theorem}\label{th:nn2}
	Let $P$ be a subset of $\discrete^D$ of size $n$.
	Then we can compute in $2^{O(D)} n \log n$ time an AVD of $P$ with $2^{O(D)}n$ regions,
	and $2^{O(D)}$  representative points per region. 
	Using this diagram, given a query point $q \in \poincare^D$,
	a point $n_2(q)\in P$ such that $d_2(q,n_2(q))$ is minimum can be returned 
	in $2^{O(D)}+O(\log n)$ time.
	This point satisfies $d_1(q,n_2(q))  \leq d_1(q,n_1(q))+2$, 
	where $n_1(q)$ is a point in $P$ such that $d_1(q,n_1(q))$ is minimum.
\end{theorem}

The discrete AVD presented above can also be turned into an AVD for a set of points $P$
in $H^{D}$, by constructing the $d_2$-AVD for the points $\{b(p): p \in P\}$,
and replacing each point $p$ in a Voronoi region by the whole box $\cell(p)$, and thus
we obtain a partition of $[0,1]^{D-1} \times (0,2] \subset \poincare^D$.
Then for a query point $q \in H^D$, we return a  point in $P$ such
that $b(p)=n_2(b(q))$. 
Theorem~\ref{th:embedmain} implies that this point is an
$O(\log D)$-additive approximate near neighbor:
\begin{corollary}\label{cor:nnh}
	Let $P$ be a subset of $\poincare^D$ of size $n$. 
	Then we can compute in $2^{O(D)}n \log n$ time an AVD of $P$ with $2^{O(D)}n$ regions,
	and $2^{O(D)}$  representative points per region. 
	Using this diagram,	for any query point $q$,
	we can return in $2^{O(D)}+O(\log n)$ time
	a point $r \in P$ such that $d_H(q,r)\leq d_H(q,n_H(q))+O(\log D)$,
	where $n_H(q)$ is a the point in $P$ such that $d_H(q,n_H(q))$ is minimum.
\end{corollary}

\bibliographystyle{plainurl}
\bibliography{hyperbolic}

\begin{thebibliography}{10}

\bibitem{AlthoferDDJS93}
Ingo Alth{\H{o}}fer, Gautam Das, David~P. Dobkin, Deborah Joseph, and
  Jos{\'{e}} Soares.
\newblock On sparse spanners of weighted graphs.
\newblock {\em Discret. Comput. Geom.}, 9:81--100, 1993.
\newblock \href {https://doi.org/10.1007/BF02189308}
  {\path{doi:10.1007/BF02189308}}.

\bibitem{AryaDMSS95}
Sunil Arya, Gautam Das, David Mount, Jeffrey Salowe, and Michiel Smid.
\newblock Euclidean spanners: short, thin, and lanky.
\newblock In {\em Proc. 27th {ACM} Symposium on Theory of Computing}, pages
  489--498, 1995.
\newblock \href {https://doi.org/10.1145/225058.225191}
  {\path{doi:10.1145/225058.225191}}.

\bibitem{Cannon97}
James Cannon, William Floyd, Richard Kenyon, and Walter Parry.
\newblock Hyperbolic geometry.
\newblock In Silvio Levy, editor, {\em Flavors of Geometry}, volume~31, pages
  167--196. MSRI Publications, 1997.

\bibitem{ChepoiDEHVX12}
Victor Chepoi, Feodor~F. Dragan, Bertrand Estellon, Michel Habib, Yann
  Vax{\`{e}}s, and Yang Xiang.
\newblock Additive spanners and distance and routing labeling schemes for
  hyperbolic graphs.
\newblock {\em Algorithmica}, 62(3-4):713--732, 2012.
\newblock \href {https://doi.org/10.1007/s00453-010-9478-x}
  {\path{doi:10.1007/s00453-010-9478-x}}.

\bibitem{GaneaBH18}
Octavian{-}Eugen Ganea, Gary B{\'{e}}cigneul, and Thomas Hofmann.
\newblock Hyperbolic neural networks.
\newblock In Samy Bengio, Hanna~M. Wallach, Hugo Larochelle, Kristen Grauman,
  Nicol{\`{o}} Cesa{-}Bianchi, and Roman Garnett, editors, {\em Advances in
  Neural Information Processing Systems 31: Annual Conference on Neural
  Information Processing Systems 2018, NeurIPS 2018}, pages 5350--5360, 2018.

\bibitem{Gromov1987}
M.~Gromov.
\newblock {\em Hyperbolic Groups}, chapter~6, pages 75--263.
\newblock Springer New York, 1987.

\bibitem{HPbook}
Sariel Har-peled.
\newblock {\em Geometric Approximation Algorithms}.
\newblock American Mathematical Society, 2011.

\bibitem{kisfaludibak}
S{\'a}ndor Kisfaludi-Bak.
\newblock {A Quasi-Polynomial Algorithm for Well-Spaced Hyperbolic TSP}.
\newblock In {\em Proc. 36th International Symposium on Computational Geometry
  (SoCG 2020)}, pages 55:1--55:15, 2020.

\bibitem{kisfaludibak2023quadtree}
Sándor Kisfaludi-Bak and Geert van Wordragen.
\newblock A quadtree for hyperbolic space, 2023.
\newblock \href {https://arxiv.org/abs/2305.01356} {\path{arXiv:2305.01356}}.

\bibitem{KL06}
Robert Krauthgamer and James~R. Lee.
\newblock Algorithms on negatively curved spaces.
\newblock In {\em 2006 47th Annual IEEE Symposium on Foundations of Computer
  Science (FOCS'06)}, pages 119--132, 2006.
\newblock \href {https://doi.org/10.1109/FOCS.2006.9}
  {\path{doi:10.1109/FOCS.2006.9}}.

\bibitem{Raskhodnikova10}
Sofya Raskhodnikova.
\newblock Transitive-closure spanners: {A} survey.
\newblock In Oded Goldreich, editor, {\em Property Testing - Current Research
  and Surveys}, volume 6390 of {\em Lecture Notes in Computer Science}, pages
  167--196. Springer, 2010.
\newblock \href {https://doi.org/10.1007/978-3-642-16367-8\_10}
  {\path{doi:10.1007/978-3-642-16367-8\_10}}.

\bibitem{1354510}
Y.~Shavitt and T.~Tankel.
\newblock On the curvature of the internet and its usage for overlay
  construction and distance estimation.
\newblock In {\em IEEE INFOCOM 2004}, volume~1, page 384, 2004.
\newblock \href {https://doi.org/10.1109/INFCOM.2004.1354510}
  {\path{doi:10.1109/INFCOM.2004.1354510}}.

\bibitem{Thorup97}
Mikkel Thorup.
\newblock Parallel shortcutting of rooted trees.
\newblock {\em J. Algorithms}, 23(1):139--159, 1997.
\newblock \href {https://doi.org/10.1006/jagm.1996.0829}
  {\path{doi:10.1006/jagm.1996.0829}}.

\end{thebibliography}

\end{document}